\documentclass[authoryear,preprint,11pt]{elsarticle}
\usepackage{amssymb}
\usepackage{eurosym}
\usepackage{color,colortbl}
\usepackage{bm}
\usepackage{booktabs}
\usepackage{multirow}
\usepackage{longtable}
\usepackage{arydshln}
\usepackage{comment}

\usepackage{centernot}

\usepackage[a4paper, total={8in, 8in}]{geometry}
\usepackage{subcaption}
\usepackage{tikz}
\usepackage{pgfplots}
\usetikzlibrary{arrows,calc}
\usepackage{graphicx}
\usepackage{fancybox}
\usepackage{pifont}

\usepackage{multirow}

\usepackage[inline]{enumitem}
\usepackage{xcolor}

\usepackage{etoolbox}
\makeatletter
\patchcmd{\ps@pprintTitle}{\footnotesize\itshape
Preprint submitted to \ifx\@journal\@empty Elsevier
\else\@journal\fi\hfill\today}{\relax}{}{}
\makeatother

\usepackage{enumitem}

\usepackage{amsthm,amsmath,amssymb,eurosym}
\newtheorem{theorem}{Theorem}

\newtheorem{lemma}{Lemma}
\newtheorem{definition}{Definition}
\newtheorem{remark}{Remark}
\newtheorem{proposition}{Proposition}

\newtheorem{condition}{Condition}

\usepackage{graphicx}
\newcommand{\starcupN}{\square\kern-0.58em\mbox{\scriptsize{$?$}}}

\journal{European Journal of Operational Research}

\usepackage{lineno, blindtext}

\begin{document}

\begin{frontmatter}
    \title{{\sc{Electre-Score}}: A first outranking based method for scoring actions}
    \author[ist]{Jos\'e Rui {\sc Figueira}\corref{cor1}}\ead{figueira@tecnico.ulisboa.pt}
    \author[deb,corl]{Salvatore {\sc Greco}}
    \author[lamsade]{Bernard {\sc Roy}}
    \address[deb]{Department of Economics and Business, University of Catania, Catania, Italy}
    \address[lamsade]{Universit\'e Paris-Dauphine, PSL Research University, CNRS (UMR 7243), LAMSADE, Paris, France}
    \address[ist]{CEG-IST, Instituto Superior T\'{e}cnico,  Universidade de Lisboa, Lisbon, Portugal}
    \address[corl]{Portsmouth Business School, Centre of Operations Research and Logistics (CORL), \\ University of Portsmouth,  Portsmouth, United Kingdom}
    \cortext[cor1]{Corresponding author at: CEG-IST, Instituto Superior T\'ecnico, Universidade de Lisboa, Av. Rovisco Pais 1,
    1049-001, Lisboa, Portugal. Phone: +351 21 841 77 29. Fax: +351 21 841 79 79. Date: May 15, 2019.}

    \begin{abstract}
        \noindent  In this paper we present (to the best of our knowledge) the first outranking method to assign a score to each alternative. It is a method of the {\sc{Electre}} family, and we will call it {\sc{Electre-Score}}. Contrarily to the Multi-Attribute Value Theory methods, {\sc{Electre-Score}} does not construct a value function for each criterion, and then proceeds to the aggregation into a single value. It rather, makes use of the outranking relations to make a comparison with reference sets of actions (to which a score is assigned) and proposes a score range to each alternative, instead of a single value. This is a more robust way of proceeding given the fragility of a single score. The sets of limiting profiles are defined with the help of {\sc{Electre Tri-nB}} and the reference scores are assigned to them through the application of the deck of cards technique. The fact of being able to use outranking relations, makes it also possible to take into account the imperfect knowledge of data and avoids systematic compensatory effects. Some fundamental theoretical results guaranteing the consistency of the method and an illustrative example are also provided in this paper.
    \end{abstract}
    \vspace{0.25cm}
\end{frontmatter}


\section{Introduction}\label{sec:Introduction}
\noindent Multiple criteria decision aiding (MCDA) is a discipline that comprises methods and techniques used to produce information in order to enable decision-makers to make better and informed decisions. Over the last years there was a tremendous growth in the development of new methods, strengthening the maturity of the existing ones, and in increasing their application to deal with real-world decision aiding situations of a crucial importance for organizations. One of the most magnificent and wonderful features of these methods is their diversity \citep[see][for a vast panoply od methods]{GrecoEtAl16}.

The three major families of methods can be divided by taking into account the form of the results provided to the decision-makers. There are scoring based methods, which assign a score to each action or alternative \citep{KeeneyRa93}, the outranking based methods that provide one or several outranking relations between ordered pairs of alternatives and exploit such relations \citep{RoyBo93}, and there are also rule based approaches wich provide the decision-makers with a set of decision rules \citep{GrecoEtAl01}.

{\sc{Electre}} methods \citep{FigueiraEtAl16} play a central role in the family of outranking approaches. Since their inception in the middle of the sixties of last century they were object of several studies, extensions, generalizations, new developments, and many applications in real-world situations. For a comprehensive survey of these methods see \cite{GovindanJe16}. The most relevant features of {\sc{Electre}} methods that make them adequate to deal with several situations are the following \citep{FigueiraEtAl13}: (1) they can deal with both the quantitative and the qualitative nature of the criteria scales; (2) the scales can be of very heterogeneous types (meters, noisy, delay, costs, return, etc); (3) the compensatory effects are not relevant in a systematic way (this is mainly due to the use of the concordance index and the existence of veto thresholds that avoid some compensability); (4) they are able to take into account the imperfect knowledge of data (uncertainty, imprecision, and ill-determination) and the arbitrariness when building the criteria; and, finally, (5) they are very adequate to take into account the reasons for and the reasons against an outranking.

Of course, as all the MCDA methods, the ones from {\sc{Electre}}  family are not perfect and also suffered from some drawbacks: (1) the intransitivity phenomenon may occur (and even worse, it may be quite frequent); (2) the phenomenon of the dependence with respect to irrelevant alternatives may be present; (3) if all the criteria are of a quantitative nature and no imperfect knowledge and arbitrariness are present, and in addition the decision-maker allows for a systematic compensation, we can do better with other methods, namely the ones of Multi-Attribute Value Theory (MAVT) family; and (4) if it is necessary assigning a score to each action, these methods are note adequate at all.

In reality the compensation between criteria is not always allowed by decision-makers and imperfect knowledge and arbitrariness are often present when dealing with practical decision aiding situations. Since it will be almost impossible to avoid the two other phenomena (intransitivities and dependence), would it be possible to build a scoring based {\sc{Electre}} method? This is the challenge we would like to face and the objective of this paper is thus to propose such a method.

The new method consists of following main steps:

\begin{enumerate}
  \item Several reference sets of actions are built. We used the limiting profiles as in {\sc{Electre-Tri-nB}} \citep{FernandezEtAl17} as our reference sets. Such sets of reference actions must fulfill some important separability conditions. Note that, whichever is the procedure to define \emph{a priori} the set of reference actions, it must satisfy a certain number of conditions that should be as weak as possible in order the method to be fruitfully applied in a vast generality of cases. To start, we consider rather restrictive conditions that must be relaxed as much as possible in the following but only if these less restrictive conditions would maintain the validity of the method.
  \item With a deck of cards technique, we can  assign a value to each set of limiting profiles after choosing two reference values. This is a similar technique as the one proposed in \cite{BotteroEtAl18} for building interval scales.
  \item The last step consists of comparing each action to the reference sets and assigning a scoring interval to it.
\end{enumerate}

It should be remarked that the nature of this scoring method is different from those belonging to the MAVT family, since we do not consider a transformation of the scale for each criterion into a partial value function, and then aggregate all the partial values into a single one. The score obtained with our procedure supplies a comprehensive evaluation without using a procedure for converting each criterion into a value function.

The paper is organized as follows. Section 2 is devoted to some fundamental concepts, their definitions, and corresponding notation. Section 3 presents the new {\sc{Electre-Score}} method, including the conditions for the construction of the reference set, the assignment of a scoring range to each action, and the formal definitions of the lower and upper bounds of such a range. Section 4 is related to the conditions about the set of limiting profiles that allow the procedure to be in accordance with the objectives. Section 5 is devoted to the theoretical results proving desirable properties of the procedure. Section 6 presents the method by means of an illustrative example along with some practical aspects. Finally, the last section provides the main conclusions and some lines for future research.

\section{Concepts, definitions, and notation}\label{sec:Concepts}

\noindent To start we need to introduce a few notation. Let $A = \{a_1,\ldots,a_i,\ldots,a_m\}$ denote the set of actions to which an interval score must be assigned to each of them, $G = \{g_1,\ldots,g_j,\ldots,g_n\}$ denote the set of criteria used to assess the performance of such actions, and  $g_j(a_i)$ denotes the performances of action $a_i$ on criterion $g_j$ (with all the performances we can build a performance table). Consider also the collection of sets of references actions $B = \{B_{x_1},\ldots,B_{x_k},\ldots,B_{x_\ell}\}$, for which a score is previously defined, and let $X = \{x_1,\ldots,x_k,\ldots,x_\ell\}$ the set of such scores. Each set is composed of at least one limiting profile as in {\sc{Electre Tri-nB}} \citep[see][]{FernandezEtAl17}, i.e., $B_{x_k} = \{b_{k1}, \ldots, b_{kp}, \ldots, b_{kp_k}\}$. As part of {\sc{Electre}} methods a credibility degree between all ordered pairs of actions, $\sigma(a,b)$, must be computed. This credibility measures on a scale $[0,1]$ the degree in which action $a$ outranks action $b$. In order to pass from a fuzzy relation to a crisp one, we need to define what is called the  cutting level $\lambda \in ]0.5,1]$. A brief description of the way {\sc{Electre}} methods compute the degree of credibility is available in the Appendix of this paper.

The next three definitions are fundamental.

\begin{definition}[Dominance]\label{def:Dominance}
    Consider two actions $a$ and $b$. Action $a$ \emph{dominates} action $b$, whenever $g_j(a) \geqslant g_j(b)$, for all $j=1,\ldots,n$,  with at least one strict inequality. Let $a\Delta b$ denote such a binary dominance relation.
\end{definition}

\begin{definition}[Fundamental outranking binary relation]\label{def:BinaryOutrakingRelations} Consider two actions $a$ and $b$ from set $A$. Once fixed the cutting level $\lambda$, we say that action $a$ outranks (or is at least as good as) action $b$, denoted by $a\succsim^\lambda b$ iff $\sigma(a,b) \geqslant \lambda$. It is easy to see that $\succsim^\lambda$ is a reflexive, but not necessarily symmetric and transitive binary relation. In what follows and with some abuse of the mathematical language, we will use simply $a\succsim b$ instead of $a\succsim^\lambda b$ for denoting this $\lambda-$outranking binary relation (the same applies for the binary relations introduced in next definition).
\end{definition}

\begin{definition}[Derived binary relations]\label{def:DerivedBinaryRelations} From the fundamental outranking binary relation $\succsim$, we can derive, for two actions $a,b \in A$, the following three binary relations (which correspond to all possible combinations of the presence and non-presence of an outranking relation between $a$ and $b$, and $b$ and $a$, respectively).
    \begin{enumerate}[label=\roman*)]
        \item  $a\succ b$ (\emph{preference} in favor of $a$, which means that $a$ is preferred to $b$) iff $a\succsim b$ and $\mbox{not}(b\succsim a)$;
        \item  $b\succ a$ (\emph{preference} in favor of $b$, which means that $b$ is preferred to $a$) iff $b\succsim a$ and $\mbox{not}(a\succsim b)$;
        \item  $a\sim b$ (\emph{indifference}, which means that actions, $a$ and $b$, are indifferent) iff $a\succsim b$ and $b\succsim a$;
        \item  $a\parallel b$ (\emph{incomparability}, which means that actions, $a$ and $b$, are incomparable) iff $\mbox{not}(a\succsim b)$ and $\mbox{not}(b\succsim a)$.
    \end{enumerate}
    ($\succ$ is irreflexive and asymmetric; $\sim$ is reflexive and symmetric; and, $\parallel$ is irreflexive and symmetric.)
\end{definition}

\begin{remark}\label{rem:PropertiesBinaryRelations}
    From Definitions \ref{def:BinaryOutrakingRelations} and \ref{def:DerivedBinaryRelations}, it is easy to see that $a\succsim b$ implies, either $a\succ b$ or $a\sim b$. Taking into account the dominance relation of Definition \ref{def:Dominance}, the following properties hold.
    \begin{subequations}
    \renewcommand{\theequation}{\theparentequation.\arabic{equation}}
        \begin{align}
            a \Delta b                                 \Rightarrow  a\succsim b \label{eqn:Rem1_1}\\
            a \succsim b \; \mbox{and} \; b \Delta c   \Rightarrow  a\succsim c \label{eqn:Rem1_2}\\
            a \Delta b   \; \mbox{and} \; b \succsim c \Rightarrow  a\succsim c \label{eqn:Rem1_3}\\
            a \succ b    \; \mbox{and} \; b \Delta c   \Rightarrow  a\succ c    \label{eqn:Rem1_4}\\
            a \Delta b   \; \mbox{and} \; b \succ c    \Rightarrow  a\succ c    \label{eqn:Rem1_5}
        \end{align}
    \end{subequations}
\end{remark}

\section{{\sc{Electre-Score}}}\label{sec:ElectreScore}
\noindent This section provides the basic foundations of the {\sc{Electre-Score}} method, i.e., the necessary elements for the construction of a reference set, the conditions needed for assigning a score range to each action, and the formal definition of the lower and upper bounds of such a range. Most of the material presented in this section is closely related to the {\sc{Electre Tri-nB}} method \citep[see][]{FernandezEtAl17}.

    \subsection{Constructing a reference set}\label{sec:ReferenceSet}
    \noindent The definition of the reference set as well as the basic assumption with respect to such reference set are presented next. The subsection also provided some more results in the same line as in \cite{FernandezEtAl17}.

    \begin{definition}[Set of reference actions]\label{def:ReferenceActions}
        Let $X=\{x_1,\ldots,x_k,\ldots,x_\ell\}$ denote the set of values considered as references scores, and $B_{x_k}=
        \{b_{k1}, \ldots, b_{kp}, \ldots, b_{kp_k}\}$ denote the set of reference actions used to characterize score $x_k$. As
        a result $B = \bigcup_{k}^{\ell}B_{x_k}$ denotes a set containing all the reference actions.
    \end{definition}

    \begin{condition}[Basic assumptions]\label{cond:BasicAssumptions}
        The score $x_k$ is characterized by a set of reference actions, $B_{x_k} = \{b_{x_{k}1}, \ldots, b_{x_{k}p}, \ldots, b_{x_{k}p_k}\}$, for $k = 1,\ldots,\ell$, such that:
        \begin{enumerate}[label=\roman*)]
            \item For all  $b_{kp}, b_{kq} \in B_{x_k}$ there is no preference between $b_{kp}$ and $b_{kq}$ (this implies, there is only the possibility to have either $b_{kp} \sim b_{kq}$ or $b_{kp} \parallel b_{kq}$);
            \item For all  $b_{kp} \in B_{x_k}$ and $b_{hq} \in B_{x_h}$ ($x_k > x_h$), it is not possible to have $b_{hq} \succ b_{kp}$.
        \end{enumerate}
    \end{condition}

    \begin{definition}[Relations between an action and a reference set]\label{def:RelationsActionsSets}
        Consider the following relations between an action $a$, and a set of reference actions, $B_{x_k}$ \citep[see][]{FernandezEtAl17}.
        \begin{enumerate}[label=\roman*)]
            \item $a\succsim B_{x_k}$ iff, for all $b_{kq} \in B_{x_k}$, either $a\parallel b_{kq}$ or $a\succsim b_{kq}$, the latter relation being fulfilled by at least one $b_{kq} \in B_{x_k}$ (note that, for all $b_{kq} \in B_{x_k}$, it is not possible to have $b_{kq}\succ a$);
            \item $B_{x_k}\succsim a$ iff, for all $b_{kq} \in B_{x_k}$, either $b_{kq}\parallel a$ or $b_{kq}\succsim a$, the latter relation being fulfilled by at least one $b_{kq} \in B_{x_k}$ (note that, for all $b_{kq} \in B_{x_k}$, it is not possible to have $a\succ b_{kq}$);
            \item $a\succ B_{x_k}$ iff, for all $b_{kq} \in B_{x_k}$, either $a\parallel b_{kq}$ or $a\sim b_{kq}$, or $a\succ b_{kq}$, the latter relation being fulfilled by at least one $b_{kq} \in B_{x_k}$ (note that, for all $b_{kq} \in B_{x_k}$, it is not possible to have $b_{kq}\succ a$);
            \item $B_{x_k}\succ a$  iff, for all $b_{kq} \in B_{x_k}$, either $b_{kq}\parallel a$ or $b_{kq}\sim a$, or $b_{kq}\succ a$, the latter relation being fulfilled by at least one $b_{kq} \in B_{x_k}$ (note that, for all $b_{kq} \in B_{x_k}$, it is not possible to have $a\succ b_{kq}$);
            \item $a\sim B_{x_k}$ iff, for all $b_{kq} \in B_{x_k}$, either $a\parallel b_{kq}$ or $a\sim b_{kq}$, the latter relation being fulfilled by at least one $b_{kq} \in B_{x_k}$ (note that, since $\sim$ is symmetric, for all $b_{kq} \in B_{x_k}$, it is not possible to have $b_{kq}\succ a$ or $a\succ b_{kq}$);
            \item $a\parallel B_{x_k}$ iff, for all $b_{kq} \in B_{x_k}$, either $a\parallel b_{kq}$ or, when $a\succ b_{kq}$, for some $b_{kq} \in B_{x_k}$, $b_{kp}\succ a$, for some $b_{kp} \in B_{x_k}$, with $b_{kq}\neq b_{kp}$ (note that, since $\sim$ is symmetric, it is not possible to have $B_{x_k}\succsim a$ or $a\succsim B_{x_k}$).
        \end{enumerate}
    \end{definition}

    \begin{remark}\label{rem:PropertiesRelationsSetAction}
        The following implications can be derived from Definitions \ref{def:Dominance} and \ref{def:RelationsActionsSets} \citep[see][]{FernandezEtAl17}.
        \begin{enumerate}[label=\roman*)]
            \item $a\succ B_{x_k}$ implies $a\succsim B_{x_k}$;
            \item $a\succ B_{x_k}$ implies $\mbox{not}(B_{x_k}\succsim a)$, and consequently $a\succ B_{x_k}$ also implies $\mbox{not}(B_{x_k}\succ a)$;
            \item  $B_{x_k}\succ a$ implies $\mbox{not}(a\succsim B_{x_k})$, and consequently $B_{x_k}\succ a$ also implies $\mbox{not}(a\succ B_{x_k})$;
            \item $B_{x_k}\succ a$ implies $B_{x_k}\succsim a$;
            \item $B_{x_k}\succ a$ and $a\Delta b$ implies $B_{x_k}\succ b$;
            \item $a\Delta b$ and $b\succsim B_{x_k}$ implies $a\succsim B_{x_k}$;
            \item $a\Delta b$ and $b\succ B_{x_k}$ implies $a\succ B_{x_k}$.
        \end{enumerate}
    \end{remark}

    \subsection{Assigning a score range to each action}\label{sec:Scores}
    \noindent The aim of the ELECTRE score method we are presenting is the identification of a range, $]s^l(a), s^u(a)[$ for the score $s(a)$ to be assigned to actions $a \in A$.  Let us first discuss some conditions that guarantee the existence of such a range $]s^l(a), s^u(a)[$. Since we only know the relations between $a$ and the elements of set $B$ (see Definition \ref{def:RelationsActionsSets}), the lower bound, $s^l(a)$, and the upper bound, $s^u(a)$, of the range cannot be fixed \emph{a priori}. However, it is easy to see that we cannot have $B_{x}\succ a$, for any score $x$ lower than or equal to the lower bound, $s^l(a)$. Otherwise, the lower bound was not properly defined and we could move down (decrease) its value, which makes no sense. Analogously, it is easy to see that we cannot have $a \succ B_{x}$ for any score $x$ greater than or equal to the upper bound, $s^u(a)$. Otherwise, the upper bound was not properly defined and we could move up (increase) its value, which makes no sense. This reasoning led us to establish the following two necessary conditions for the existence of the range  $]s^l(a), s^u(a)[$.

    \begin{condition}[Lower bound necessary condition]\label{cond:LowerBoundNecCond}
        If $x \leqslant s^l(a)$, then $\mbox{not}(B_{x}\succ a)$, for all $a \in A$.
    \end{condition}

    \begin{condition}[Upper bound necessary condition]\label{cond:UpperBoundNecCond}
         If $x \geqslant s^u(a)$, then $\mbox{not}(a\succ B_{x})$, for all $a \in A$.
    \end{condition}

    These two conditions are necessary for the existence of the range, but they are not sufficient, since we know nothing about the relations between $B_x$ and $a$, for an $x$ value strictly comprised within the range $]s^l(a), s^u(a)[$. However, it is easy to see that we cannot have $B_{x}\succ a$, for any score $x$ strictly greater than the lower bound, $s^l(a)$. Otherwise, the lower bound was not properly defined and we could move up (increase) its value, which makes no sense. Analogously, it is easy to see that we cannot have $a\succ B_{x}$, for any score $x$ strictly lower than the upper bound, $s^u(a)$. Otherwise, the upper bound was not properly defined and we could move down (decrease) its value, which makes no sense. This reasoning led us to establish the following two sufficient conditions for the existence of the range  $]s^l(a), s^u(a)[$.

    \begin{condition}[No active preference condition of reference sets in the score range]\label{cond:LowerBoundSufCond}
        If $s^l(a) < x < s^u(a)$, then $\mbox{not}(B_{x}\succ a)$, for all $a \in A$.
    \end{condition}

    If this condition was violated, no score $s(a)$, such that $s^l(a) \leqslant s(a) < s^u(a)$,  could be justified.

    \begin{condition}[No passive preference condition of reference sets in the score range]\label{cond:UpperBoundSufCond}
         If $s^l(a) < x < s^u(a)$, then $\mbox{not}(a\succ B_{x})$, for all $a \in A$.
    \end{condition}

    Analogously, if this condition was violated, no score $s(a)$, such that $s^l(a) < s(a) \leqslant s^u(a)$,  could be justified.

    It is easy to see that Conditions \ref{cond:LowerBoundSufCond} and \ref{cond:UpperBoundSufCond} are fulfilled if and only if the next (equivalent) condition holds.

    \begin{condition}[No preference condition of reference sets in the interval]\label{cond:UpperBoundCond}
        For all possible scores, $x$, such that $s^l(a) < x < s^u(a)$, $\mbox{not}(B_{x}\succ a)$ and $\mbox{not}(a \succ B_{x})$, for all $a \in A$.
    \end{condition}

    This condition means that assigning to action $a$, any feasible score, i.e, any score $x$ such that $s^l(a) < x < s^u(a)$, is only possible when, either $a \sim B_x$ or $a\parallel B_x$.

    Observe that Conditions \ref{cond:LowerBoundNecCond} and \ref{cond:LowerBoundSufCond} can be replaced by the following condition.

    \begin{condition}[Lower bound general condition]\label{cond:NonPreferenceActionCond}
         If  $x > s^l(a)$, then $\mbox{not}(a\succ B_{x})$, for all $a \in A$.
    \end{condition}

    Analogously, Conditions \ref{cond:UpperBoundNecCond} and \ref{cond:UpperBoundSufCond} can be replaced by the following condition.

    \begin{condition}[Upper bound general condition]\label{cond:NonPreferenceSetCond}
         If  $x < s^u(a)$, then $\mbox{not}(B_{x}\succ a)$, for all $a \in A$.
		\end{condition}

    \begin{remark}\label{rem:NecSufConditions}
         Conditions \ref{cond:LowerBoundNecCond} to \ref{cond:UpperBoundSufCond} do not imply by no means that the next condition is automatically fulfilled.
    \end{remark}

    \begin{condition}[Indifference/incomparability]\label{cond:NonPreferenceCond}
         If   $a\sim B_{x}$ or $a \parallel B_{x}$, then $s^l(a) < x < s^u(a)$, for all $a \in A$.
    \end{condition}

Figure 1 illustrates the above conditions. Notation $\mathbf{[Ck]}$ is used instead of Condition $k$, for $k=3,4,5,6,7,8,9$.

     \vspace{-0.5cm}

    \begin{center}
    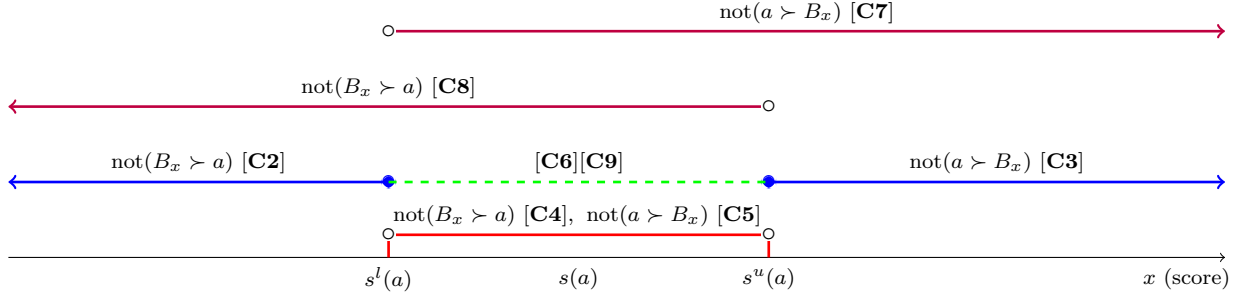
\begin{figure}[htb!]\label{fig:NecSufConditions}
    \begin{tikzpicture}[x=1cm,y=1cm,scale=1]
        \draw[->] (0,0) -- (16,0);
        \node[below] at (15.5,0) {\scriptsize $x$ (score)};
        \node[below] at (7.5,0) {\scriptsize $s(a)$};
        \draw[-,line width=1pt, red] (5,0.3) -- (10,0.3);
        \draw[-,line width=1pt, red] (5,0.3) -- (5,0);
        \draw[-,line width=1pt, red] (10,0.3) -- (10,0);
        \node[inner sep=0, fill=white] at (5, 0.3) {$\circ$};
        \node[inner sep=0, fill=white] at (10,0.3) {$\circ$};
        \node[below] at (5,0) {\scriptsize $s^l(a)$};
        \node[below] at (10,0) {\scriptsize $s^u(a)$};
        \node[] at (13,1.25) {\scriptsize $\mbox{not}(a\succ B_x)$ $\mathbf{[C3]}$};
        \draw[->,line width=1pt, blue] (10,1) -- (16,1);
        \node[inner sep=0, rounded corners, color=blue, fill=blue] at (10, 1) {$\circ$};
        \node[] at (2.5,1.25) {\scriptsize $\mbox{not}(B_x\succ a)$  $\mathbf{[C2]}$};
        \draw[<-,line width=1pt, blue] (0,1) -- (5,1);
        \node[inner sep=0, rounded corners, color=blue, fill=blue] at (5, 1) {$\circ$};
        \node[] at (6.25,0.55) {\scriptsize $\mbox{not}(B_x\succ a)$  $\mathbf{[C4]}$,};
        \node[] at (8.75,0.55) {\scriptsize $\mbox{not}(a\succ B_x)$  $\mathbf{[C5]}$};
        \draw[-,dashed, line width=1pt, green] (5,1) -- (10,1);
        \node[] at (7.5,1.25) {\scriptsize $\mathbf{[C6][C9]}$};
        \node[] at (5,2.25) {\scriptsize $\mbox{not}(B_x\succ a)$ $\mathbf{[C8]}$};
        \draw[<-,line width=1pt, purple] (0,2) -- (10,2);
        \node[inner sep=0, fill=white] at (10, 2) {$\circ$};
        \node[] at (10.5,3.25) {\scriptsize $\mbox{not}(a\succ B_x)$ $\mathbf{[C7]}$};
        \draw[->,line width=1pt, purple] (5,3) -- (16,3);
        \node[inner sep=0, fill=white] at (5, 3) {$\circ$};
    \end{tikzpicture}
    \caption{Illustration of the necessary and no active and passive preference conditions for the existence of a range for $s(a)$, $]s^l(a), s^u(a)[$}
    \end{figure}
    \end{center}

    \vspace{-0.5cm}

\subsection{Definitions of the lower and upper bounds}\label{sec:Bounds}
\noindent The formal definitions of the lower and upper bound are as follows.

    \begin{definition}[Lower bound of the score range]\label{def:LowerBound}
         The $s^l(a)$ value is the highest value, such that $a \succ B_x$ and $\mbox{not}(B_{x}\succ a)$, for all $x < s^l(a)$.
    \end{definition}

    \begin{definition}[Upper bound of the score range]\label{def:UpperBound}
         The $s^u(a)$ value is the lowest value, such that $B_x\succ a$ and $\mbox{not}(a\succ B_{x})$, for all $x > s^u(a)$.
    \end{definition}

    \begin{remark}\label{rem:OnDefinitions}
         It is obvious that Definitions \ref{def:LowerBound} and \ref{def:UpperBound} ensure that \ref{cond:LowerBoundNecCond} and \ref{cond:UpperBoundNecCond} are automatically fulfilled.
    \end{remark}

\section{Finding the conditions on $B$ allowing the procedure to be in accordance with the objectives}\label{sec:TheoreticalResults}
\noindent This section presents the conditions of the set $B$ that render the procedure or the new method coherent, i.e., in accordance with the objectives the method have been designed for. From these conditions a set of new results needed to be proved.

    \begin{condition}[Separability conditions]\label{cond:Separability}
        The following are the required separability conditions on set $B$.
        \begin{enumerate}[label={\alph{enumi}}]
            \item \emph{Dominance based separability conditions}.
                \begin{enumerate}[label={\alph{enumi}.\arabic*}]
                    \item \emph{Strong dominance}. Consider $x_h > x_k$. For any two reference actions, $b_{hq} \in B_{x_h}$ and $b_{kp} \in B_{x_k}$, the relation $b_{hq}\Delta b_{kp}$ holds.
                    \item \emph{Soft dominance}.
                        \begin{itemize}
                            \item[$a.2.p$] (\emph{primal}): Consider $x_h > x_k$. For all $b_{kp} \in B_{x_k}$, there is at least one $b_{hq} \in B_{x_h}$ such that $b_{hq}\Delta b_{kp}$.
                            \item[$a.2.d$] (\emph{dual}): Consider $x_h > x_k$. For all $b_{hq} \in B_{x_h}$, there is at least one $b_{kp} \in B_{x_k}$ such that $b_{hq}\Delta b_{kp}$.
                        \end{itemize}
                \end{enumerate}
            \item \emph{Preference based separability conditions}.
                \begin{enumerate}[label={\alph{enumi}.\arabic*}]
                    \item \emph{Strong preference}.  For any two reference actions, $b_{hq} \in B_{x_h}$ and $b_{kp} \in B_{x_k}$, the relation $b_{hq}\succ b_{kp}$ holds.
                    \item \emph{Soft preference}.
                        \begin{itemize}
                            \item[$b.2.p$] (\emph{primal}): Consider $x_h > x_k$. For all $b_{kp} \in B_{x_k}$, there is at least one $b_{hq} \in B_{x_h}$ such that $b_{hq}\succ b_{kp}$.
                            \item[$b.2.d$] (\emph{dual}): Consider $x_h > x_k$. For all $b_{hq} \in B_{x_h}$, there is at least one $b_{kp} \in B_{x_k}$ such that $b_{hq}\succ b_{kp}$.
                        \end{itemize}
                \end{enumerate}
        \end{enumerate}
    \end{condition}
    \begin{proposition}[Comparisons of the actions against reference sets]\label{prop:PropertiesActionsSets}
        ~~~~
        \begin{enumerate}[label=\roman*)]
            \item If the primal soft dominance separability condition holds for $B$, then, for all $a \in A$, $a\succsim B_{x_k}$ implies $\mbox{not}(B_{x_h}\succ a)$, for all $x_h < x_k$;
            \item If the primal soft dominance separability condition holds for $B$, then, for all $a \in A$, $B_{x_k}\succ a$ implies $\mbox{not}(a \succsim B_{x_h})$, for all $x_h > x_k$;
            \item If both the primal and the dual soft dominance separability condition hold for $B$, then, for all $a \in A$, $a\succsim B_{x_k}$ implies $(a \succsim B_{x_h})$, for all $x_h < x_k$;
            \item If both the primal and the dual soft dominance separability condition hold for $B$, then, for all $a \in A$, $B_{x_k}\succ a$ implies $(B_{x_h}\succ a)$, for all $x_h > x_k$;
            \item If the dual soft dominance separability condition holds for $B$, then, for all $b_{kp} \in B_{x_k}$, $b_{kp}\succsim B_{x_h}$, for all $x_h \leqslant x_k$;
            \item  If the primal soft dominance separability condition holds for $B$, then, for all $b_{kp} \in B_{x_k}$, $B_{x_h}\succ b_{kp}$, for all $x_h > x_k$.
        \end{enumerate}
    \end{proposition}

    \begin{proof}
      We shall prove only point $iv)$ and $vi)$ that will be used several times in the proofs of other propositions in the paper. The other points have similar proofs. Analogous results have been obtained with respect to sorting procedures in \citep{FernandezEtAl17}.
			
			If $B_{x_h}\succ a$ then there exist $b_{hp} \in B_{x_h}$ such that $b_{hp} \succ a$. Since $k >h$, by Condition $a.2.p$ (primal), there exist $b_{kq} \in B_{x_k}$ such that $b_{kq} \Delta b_{hp}$ and, consequently, $b_{kq} \succ a$. By contradiction, let us suppose that there exist $b_{kr} \in B_{x_k}$ such that $a \succ b_{kr}$. By Condition  $a.2.d$ (dual), there should exist $b_{hs} \in B_{x_h}$, such that $b_{kr} \Delta b_{hs}$. Therefore we would get $a \succ b_{hs}$, which would imply not $(B_{x_h} \succ a)$ contradicting the hypothesis $B_{x_h} \succ a$. This concludes the proof of point $iv)$.

If $a \succ B_{x_h}$ then there exist $b_{hp} \in B_{x_h}$ such that $a \succ b_{hp}$. Since $h >k$, by Condition $a.2.d$ (dual), there exist $b_{kq} \in B_{x_k}$ such that $b_{hp} \Delta b_{kq}$ and, consequently, $a \succ b_{kq}$. By contradiction, let us suppose that there exist $b_{kr} \in B_{x_k}$ such that $b_{kr} \succ a$. By Condition  $a.2.p$ (primal), there should exist $b_{hs} \in B_{x_h}$, such that $b_{hs} \Delta b_{kr}$. Therefore we would get $b_{hs} \succ a$, which would imply not $(a \succ B_{x_h})$ contradicting the hypothesis $a \succ B_{x_h}$. This concludes the proof of point $vi)$.
			    \end{proof}

    \begin{proposition}[Upper and lower bound sufficiency]\label{prop:SoftDominanceSufficiency}
         If $B$ fulfills both the primal and the dual soft dominance separability conditions (see Condition \ref{cond:Separability}, points   $a.2.p$ and $a.2.d$), then Conditions \ref{cond:LowerBoundSufCond} and \ref{cond:UpperBoundSufCond} are also fulfilled.
    \end{proposition}

    \begin{proof}
        By contradiction, suppose that there exist $x_h \in X$ and  $a \in A$ such that  $x_h < s^u(a)$ and $B_{x_h}\succ a$. By Definition \ref{def:UpperBound}, this would imply that there should exist $x_k >x_h$ such that $a \succ B_{x_k}$. Thus, there would exist $b_{kp} \in B_{x_k}$ such that $a \succ b_{kp}$. By Condition $a.2.p$ (primal), there should exist $b_{hq} \in B_{x_h}$ such that $b_{kp} \Delta b_{hq}$. Consequently, $a \succ b_{hq}$, which would imply not($B_{x_h} \succ a$), contradicting thus the hypothesis $B_{x_h}\succ a$. Therefore, Condition \ref{cond:LowerBoundSufCond} holds.
	
	    Analogously, by contradiction, suppose that there exists $x_h \in X$ and  $a \in A$ such that $x_h > s^l(a)$ and $a \succ B_{x_h}$. By Definition \ref{def:LowerBound}, this would imply that there should exist $x_k < x_h$ such that $B_{x_k} \succ a$. Thus, there would exist $b_{kp} \in B_{x_k}$ such that $b_{kp} \succ a$. By Condition $a.2.d$ (dual), there should exist $b_{hq} \in B_{x_h}$ such that $b_{hq} \Delta b_{kp}$. Consequently, $b_{hq} \succ a$ which would imply not($a \succ B_{x_h}$), contradicting thus the hypothesis $a \succ B_{x_h}$. Therefore, Condition \ref{cond:UpperBoundSufCond} holds.
    \end{proof}

    \begin{proposition}[Implications of soft dominance in a preference of an action w.r.t. a set]\label{prop:SoftDominancePreferenceSet}
         If $B$ fulfills both the primal and the dual soft dominance separability conditions (see Points   $a.2.p$ and $a.2.d$ of  Condition \ref{cond:Separability}), then, for any $a \in A$, $a\succ B_x$ for all $x \leqslant s^l(a)$.
    \end{proposition}

    \begin{proof}
        By Definition \ref{def:LowerBound}, we have $a \succ B_{s^l(a)}$, and, since  Points   $a.2.p$ and $a.2.d$ (primal and dual) of Condition \ref{cond:Separability} hold, by  Proposition \ref{prop:PropertiesActionsSets}, we obtain  $a \succ B_x$, for all  $x < s^l(a)$.
    \end{proof}

    \begin{proposition}[Implications of soft dominance in a preference of a set w.r.t. an action]\label{prop:SoftDominancePreferenceAction}
         If $B$ fulfills both the primal and the dual soft dominance separability conditions (see Points   $a.2.p$ and $a.2.d$ of condition \ref{cond:Separability}), then, for any $a \in A$, $B_x\succ a$ for all $x \geqslant s^u(a)$.
    \end{proposition}

    \begin{proof}
		 By Definition \ref{def:UpperBound}, we have $B_{s^u(a)}\succ a$ and, since  Points   $a.2.p$ and $a.2.d$ (primal
        and dual) of Condition \ref{cond:Separability} hold, by  Proposition \ref{prop:PropertiesActionsSets}, we obtain  $B_x  \succ a$, for all  $x > s^u(a)$.
    \end{proof}

    \begin{proposition}[Implications of soft dominance in the indifference/incomparability region]\label{prop:SoftDominancePreferenceAction_}
         If $B$ fulfills both the primal and the dual soft dominance separability conditions (see Condition \ref{cond:Separability}, points   $a.2.p$ and $a.2.d$), then Conditions \ref{cond:NonPreferenceCond} is fulfilled. Moreover, in such a case, the value $s^u(a)$ is simply the highest value, such that $a\succ B_x$, for $x < s^u(a)$, and $s^l(a)$ is the lowest value, such that $B_x \succ a$, for $x > s^l(a)$, for all $a \in A$ (i.e., the second parts of Definitions \ref{def:LowerBound} and \ref{def:UpperBound} can be neglected since they are automatically fulfilled by themselves)
    \end{proposition}

    \begin{proof}
    We have already noted that Conditions \ref{cond:LowerBoundSufCond} and \ref{cond:UpperBoundSufCond} are equivalent to Condition \ref{cond:UpperBoundCond}, in which, for all $a \in A$ and, for all values $x$ such that $s^l(a)<x<s^u(a)$, not($B_x \succ a$) and not($a \succ B_x$). Therefore, we have to prove only that if Points  $a.2.p$ and $a.2.d$ of Condition \ref{cond:Separability} are satisfied, then for all values $x$ such that not($B_x \succ a$) and not($a \succ B_x$), we have $s^l(a)<x<s^u(a)$. This is true because
		\begin{itemize}[label={--}]
			\item by Proposition \ref{prop:SoftDominancePreferenceSet}, if $x \leqslant s^l(a)$, then $a\succ B_x$; and,
			\item by Proposition \ref{prop:SoftDominancePreferenceAction}, if $x \geqslant s^u(a)$, then $B_x \succ a$.
		\end{itemize}
        Consequently, we obtain that, under the same conditions, $s^l(a)$ is simply the highest value $x$ such that $a \succ B_x$, and, $s^u(a)$ is simply the lowest value $x$ such that $B_x \succ a$.
    \end{proof}

\section{Theoretical results}\label{sec:TheoreticalResults}
\noindent This  section has the purpose of showing the consistency of the method with respect to some fundamental and quite natural requirements, namely the ones of uniqueness, independence, monotonicity, conformity, homogeneity, and stability.

    \begin{definition}[Inserting and deleting operations]\label{def:InsertionDeletingOperations} The following two operations are considered:
        \begin{enumerate}[label=\roman*)]
            \item Inserting operation:
                 \begin{enumerate}[label=i.\alph*)]
                    \item a new set $B_x$ is inserted in $B$;
                    \item a new action $b$ is inserted in $B_x$.
                 \end{enumerate}
            \item Deleting operation.
                 \begin{enumerate}[label=ii.\alph*)]
                    \item a set $B_x$ is removed from $B$;
                    \item an action $b$ is removed from $B_x$.
                 \end{enumerate}
        \end{enumerate}
    \end{definition}

    \begin{definition}[Structural requirements]\label{def:StructuralRequirements} The following structural requirements of the method are considered:
        \begin{enumerate}[label=\roman*)]
            \item \emph{Uniqueness}. For each action $a\in A$ there is a single value $s^l(a)$ and a single value $s^u(a)$.
            \item \emph{Independence}. The definition of the values $s^l(a)$ and $s^u(a)$, for action $a \in A$, does not depend on the other actions in $A\setminus\{a\}$.
            \item \emph{Monotonicity}. If $a \Delta a^\prime$, then $s^l(a) \geqslant s^l(a^\prime)$ and $s^u(a) \geqslant s^u(a^\prime)$.
            \item \emph{Conformity}. If $a \in B_{x_k}$, then $s^l(a) = x_{k-1}$ and $s^u(a) = x_{k+1}$.
            \item \emph{Homogeneity}. If two actions, $a$ and $a^\prime$, compare the same way with respect to the reference actions in $B$, then $s^l(a) = s^l(a^\prime)$ and $s^u(a) = s^u(a^\prime)$.
            \item \emph{Stability}. For $a \in A$ with $x_r=s^l(a)<s^u(a)=x_s$, let $s^{\ast l}(a)$ and $s^{\ast u}(a)$ be the new bounds  in consequence of one of the operations of Definition  \ref{def:InsertionDeletingOperations}. Then
						$$x_{r-1} \leqslant s^{\ast l}(a) \leqslant x_{r+1}$$
						and
						$$x_{s-1} \leqslant s^{\ast u}(a) \leqslant x_{s+1}.$$
        \end{enumerate}
    \end{definition}
The latter requirement means that single inserting or deleting operations according to Definition \ref{def:InsertionDeletingOperations} imply a minimal perturbation on the score range of each action.

		Note that uniqueness, independence, monotonicity and homogeneity are clearly satisfied. We only focus our attention on conformity and stability.

    \begin{theorem}[Conformity]\label{theo:Conformity}
         If $B$ fulfills both the primal and the dual soft dominance and preference separability conditions (see Points $a.2.p$, $a.2.d$, $b.2.p$, and $b.2.d$ of Condition \ref{cond:Separability},), then the conformity property of Definition
         \ref{def:StructuralRequirements} is fulfilled.
    \end{theorem}

    \begin{proof}
     By Points $b.2.p$ and $b.2.d$ of Condition \ref{cond:Separability}, for each $a \in B_{x_k}$ there exists $b_{k-1p} \in B_{x_{k-1}}$, such that
	\begin{equation}
        a \succ b_{k-1p} \label{(i)}%
    \end{equation}
    and $b_{k+1q} \in B_{x_{k+1}}$, such that
	\begin{equation}
        b_{k+1q} \succ a \label{(ii)}.%
    \end{equation}

    By contradiction, let us suppose that there exists $b_{k-1r} \in B_{x_{k-1}}$, such that
	\begin{equation}
        b_{k-1r} \succ a. \label{(iii)}%
    \end{equation}
     By Point $a.2.p$ of Condition \ref{cond:Separability}, there would exist $b_{ks} \in B_{x_k}$ such that $b_{ks} \Delta b_{k-1r}$ that, together with (\ref{(iii)}), would imply $b_{ks} \succ a$, which is impossible because we cannot have a limiting profile of a class preferred to another limiting profile of the same class. Consequently, there cannot be any profile from $B_{x_{k-1}}$ preferred to $b_{ks}$ which, together with (\ref{(i)}), implies that
	\begin{equation}
        a \succ B_{x_{k-1}}. \label{(iv)}%
    \end{equation}
    Again, by contradiction, let us suppose that there exists $b_{k+1t} \in B_{x_{k+1}}$, such that
	\begin{equation}
        a \succ b_{k+1t}. \label{(v)}%
    \end{equation}
     By Point $b.2.d$ of Condition \ref{cond:Separability}, there would exist $b_{ku} \in B_{x_k}$ such that $b_{k+1t} \Delta b_{ku}$ that, together with (\ref{(v)}) would imply $a \succ b_{ku}$, which is impossible because we cannot have a limiting profile of a class preferred to another limiting profile of the same class. Consequently, $a$ cannot be preferred to any profile from $B_{x_{k+1}}$, which, together  with (\ref{(ii)}), implies that
	\begin{equation}
        B_{x_{k+1}} \succ a. \label{(vi)}%
    \end{equation}
    Based on the above considerations and taking any $a \in B_k$, we can say that
    \begin{itemize}[label={--}]
	   \item due to (\ref{prop:PropertiesActionsSets}) and (\ref{(iv)}) for any $x_h<x_k$,
            \begin{equation}
                a \succ B_{x_h}, \label{(vii)}%
            \end{equation}
        \item from Point $i)$ of Condition 1,
            \begin{equation}
                not(a \succ B_{x_k}) \mbox{ and } not(B_{x_k} \succ a), \label{(viii)}%
            \end{equation}
        \item due to (\ref{prop:PropertiesActionsSets}) and (\ref{(vi)}) for any $x_h>x_k$,
            \begin{equation}
                B_{x_h}  \succ a.  \label{(ix)}%
            \end{equation}
            Therefore,
        \begin{itemize}
	       \item[$\circ$] $k-1$ is the maximum value of $h$, for which $a \succ B_{x_h}$ and, consequently, by Proposition
                \ref{prop:SoftDominancePreferenceAction_},  $s^l(a)=x_{k-1}$,
            \item[$\circ$] $k+1$ is the minimum value of $h$, for which, $B_{x_h} \succ a$ and, consequently, by Proposition \ref{prop:SoftDominancePreferenceAction_},  $s^u(a)=x_{k+1}$.
        \end{itemize}
    \end{itemize}
    \end{proof}

    \begin{lemma}[Inserting a reference set: consequences w.r.t. lower bound]\label{theo:InsertingSet_}
         If $B$ fulfills both the primal and the dual soft dominance separability conditions (see Points $a.2.p$ and $a.2.d$ of Condition \ref{cond:Separability}) and if a set of limiting profiles $B_x$ is inserted in $B$ in such a way that these conditions are still fulfilled, and $s^l(a)=x_s$ we have either
				$$x_s <s^{l \ast}(a) =x < x_{s+1}$$
				or
				$$s^{l \ast}(a)=s^l(a).$$
        Moreover, $s^{l \ast}(a) =x$ if and only if				
				\begin{itemize}[label={--}]
					\item $x_s < x <x_{s+1}$; and.
					\item $a \succ B_x$.
				\end{itemize}
    \end{lemma}

    \begin{proof}
        If a set $B_x$ is added we can have the following cases:
	   \begin{enumerate}
		  \item $x < s^l(a)$: in this case, by Proposition \ref{prop:SoftDominancePreferenceSet}, we obtain $a \succ B_x$.
            Consequently, the highest value $z$ from $X \cup \{x\}$ such that $a \succ B_z$ is again $s^l(a)$, which, by Proposition \ref{prop:SoftDominancePreferenceAction_}, implies $s^{l \ast}(a)=s^l(a)$;
		\item $s^l(a)=x_s < x < x_{s+1}$ and not($a \succ B_x$): also in this case the maximal value $z$ from $X \cup
           \{x\}$ such that $a \succ B_z$ is  $s^l(a)$, which, again by Proposition \ref{prop:SoftDominancePreferenceAction_}, implies $s^{l \ast}(a)=s^l(a)$;
		\item $s^l(a)=x_s < x < x_{s+1}$ and $a \succ B_x$: in this case the highest value $z$ from $X \cup \{x\}$, such
           that, $a \succ B_z$ is  $x$, which, again by Proposition \ref{prop:SoftDominancePreferenceAction_}, implies $s^{l \ast}(a)=x$;
		\item $s^l(a)=x_s < x_{s+1} < x$: in this case we cannot have $a \succ B_x$ because from Condition
            \ref{cond:NonPreferenceActionCond} (that is, by  the joint consideration of Conditions \ref{cond:LowerBoundNecCond} and \ref{cond:LowerBoundSufCond}) not($a \succ B_{x_{s+1}})$, but, by Proposition \ref{prop:PropertiesActionsSets},   $a \succ B_x$ would imply $a \succ B_{x_{s+1}}$.
	\end{enumerate}
    \end{proof}

      \begin{lemma}[Inserting a reference set: consequences for the upper bound]\label{theo:InsertingSet}
         If $B$ fulfills both the primal and the dual soft dominance separability conditions (see Points $a.2.p$ and $a.2.d$ of Condition \ref{cond:Separability}), and if a set of limiting profiles $B_x$ is inserted in $B$ in such a way that these conditions are still fulfilled, and $s^u(a)=x_t$ we have either
				$$x_{t-1} <s^{u \ast}(a) =x < x_{t}$$
				or
				$$s^{u \ast}(a)=s^u(a).$$
                Moreover, $s^{u \ast}(a) =x$ if and only if				
				\begin{itemize}[label={--}]
					\item $x_{t-1} < x <x_{t}$; and,
					\item $B_x \succ a$.
				\end{itemize}
    \end{lemma}

    \begin{proof}
    If a set $B_x$ is added we can have the following cases:
	   \begin{enumerate}
		  \item $x > s^u(a)$: in this case, by Proposition \ref{prop:SoftDominancePreferenceSet}, we obtain $B_x \succ a$.
            Consequently, the lowest value $z$ from $X \cup \{x\}$ such that $B_z \succ a$ is again $s^u(a)$, which, from Proposition \ref{prop:SoftDominancePreferenceAction_}, implies $s^{u \ast}(a)=s^u(a)$;
		\item $s^u(a)=x_t > x > x_{t-1}$ and not($B_x \succ a$): also in this case the lowest value $z$ from $X \cup
            \{x\}$ such that $B_z \succ a$ is  $s^u(a)$, which, again by Proposition \ref{prop:SoftDominancePreferenceAction_}, implies $s^{u \ast}(a)=s^u(a)$;
		\item $s^u(a)=x_t > x > x_{t-1}$ and $B_x \succ a$: in this case the lowest value $z$ from $X \cup \{x\}$, such
            that, $B_z \succ a$ is  $x$, which, again by Proposition \ref{prop:SoftDominancePreferenceAction_}, implies $s^{u \ast}(a)=x$;
		\item $s^u(a)=x_t > x_{t-1} > x$: in this case one cannot have $B_x \succ a$ because by Condition
            \ref{cond:NonPreferenceSetCond} (that is, by the joint consideration of Conditions \ref{cond:UpperBoundNecCond} and \ref{cond:UpperBoundSufCond}) not($B_{x_{t-1}} \succ a)$, but, by Proposition \ref{prop:PropertiesActionsSets},   $B_x \succ a$ would imply $B_{x_{t-1}} \succ a$.
	\end{enumerate}
    \end{proof}

  \begin{remark}[Deleting a reference set]\label{rem:DeletingSet}
         If  both the primal and the dual soft dominance separability conditions (see Points $a.2.p$ and $a.2.d$ of Condition \ref{cond:Separability}) hold for $B$ and a set of limiting profiles $B_x$ is deleted from $B$, clearly  the primal and the dual soft dominance separability conditions continue to be fulfilled.
    \end{remark}

    \begin{lemma}[Deleting a reference set: consequences w.r.t. lower bound]\label{theo:DeletingSet}
         If $B$ fulfills both the primal and the dual soft dominance separability conditions (see points $a.2.p$ and $a.2.d$ of Condition \ref{cond:Separability}), a set of limiting profiles $B_x$ is removed from $B$, and $s^l(a)=x_s$, then			
				\begin{itemize}[label={--}]
					\item $s^l(a)=x_s=x$ and $s^{l \ast}(a)=x_{s-1}$; or
					\item $s^l(a)=x_s\neq x$ and $s^{l \ast}(a)=s^l(a).$
				\end{itemize}
	\end{lemma}

    \begin{proof}
        According to Remark \ref{rem:DeletingSet}, Points $a.2.p$ and $a.2.d$ of Condition \ref{cond:Separability} hold also after $B_x$ is removed from $B$. Consequently, from Proposition \ref{prop:SoftDominancePreferenceAction_}, $s^{l\ast}(a)$ continues to be the highest value $z \in X \setminus \{x\}$ such that $a \succ B_z$. We can have two cases
		  \begin{enumerate}
		      \item[~] $s^l(a)=x_s=x$: remember from Proposition \ref{prop:SoftDominancePreferenceSet}, that the highest
                value $z \in X \setminus \{x\}$ such that $a \succ B_z$ becomes $x_{s-1}$ and, consequently, $s^{l \ast}(a)=x_{s-1}$;
		      \item[~] $s^l(a)=x_s\neq x$: the highest value $z \in X \setminus \{x\}$ such that $a \succ B_z$ continues
                 to be $x_{s}$ and, consequently, $s^{l \ast}(a)=s^l(a)=x_s$.
	\end{enumerate}
    \end{proof}

 \begin{lemma}[Deleting a reference set: consequences w.r.t. upper bound]\label{theo:DeletingSet}
         If $B$ fulfills both the primal and the dual soft dominance separability conditions (see Points $a.2.p$ and $a.2.d$ of Condition \ref{cond:Separability}), a set of limiting profiles $B_x$ is removed from $B$, and $s^u(a)=x_t$, then
				\begin{itemize}[label={--}]
					\item $s^u(a)=x_t=x$ and $s^{u \ast}(a)=x_{t+1}$; or,
					\item $s^u(a)=x_t\neq x$ and $s^{u \ast}(a)=s^u(a)$.
				\end{itemize}
				    \end{lemma}

    \begin{proof}
        Remember that from Remark \ref{rem:DeletingSet}, Points $a.2.p$ and $a.2.d$ of Condition \ref{cond:Separability} hold also after $B_x$ is removed from $B$, and, taking into account Proposition \ref{prop:SoftDominancePreferenceAction_}, $s^{u\ast}(a)$ continues to be the lowest value $z \in X \setminus \{x\}$ such that $B_z \succ a$. We can have two cases:
		\begin{enumerate}
		      \item $s^u(a)=x_t=x$: remember that from Proposition \ref{prop:SoftDominancePreferenceAction}, the lowest value $z
                  \in X \setminus \{x\}$ such that $B_z \succ a$ becomes $x_{t+1}$ and, consequently, $s^{u \ast}(a)=x_{t+1}$;
		      \item $s^u(a)=x_t\neq x$: the lowest value $z \in X \setminus \{x\}$ such that $B_z \succ a$ continues
                  to be $x_{t}$ and, consequently, $s^{u \ast}(a)=s^u(a)=x_t$.
	   \end{enumerate}
    \end{proof}

    \begin{lemma}[Inserting a reference action: consequences w.r.t. lower bound]\label{theo:InsertingSet}
         If $B$ fulfills both the primal and the dual soft dominance separability conditions (see Points $a.2.p$ and $a.2.d$ of Condition \ref{cond:Separability}), and if a limiting profile $b_{kp}$ is added to $B_{x_k}$ in such a way that these conditions are still fulfilled, and $s^l(a)=x_s$, then				
				\begin{itemize}[label={--}]
					\item $s^{l \ast}(a)=x_{s-1}$ if and only if  $b_{kp} \succ a$ and $s^l(a)=x_s=x_k$;
					\item $s^{l \ast}(a)=x_{s+1}$, if and only if $a \succ b_{kp}$, $x_k=x_{s+1}$ and there is no
                            $b_{s+1q} \in B_{x_{s+1}}$ such that $b_{s+1q} \succ a$;
					\item $s^{l \ast}(a)=s^{l}(a)$, otherwise.
				\end{itemize}
    \end{lemma}

    \begin{proof}
        If $b_{kp} \succ a$ and $s^l(a)=x_s=x_k$, after adding $b_{kp}$,we have not($a \succ B_{x_s}$). Consequently, taking into consideration Proposition \ref{prop:SoftDominancePreferenceSet}, the highest value $z$ such that $a \succ B_z$ becomes $x_{s-1}$, so that $s^{l \ast}(a)=x_{s-1}$.
	
	    If $a \succ b_{kp}$, $x_k=x_{s+1}$ and there is no other $b_{kq} \in B_{x_{s+1}}$ such that $b_{kq} \succ a$, after
        adding $b_{kp}$, we have $a \succ B_{x_{s+1}}$. In this case, the highest value $z$ such that $a \succ B_z$ is no more $x_s$ and it becomes $x_{s+1}$, in such a way that, according to Proposition \ref{prop:SoftDominancePreferenceAction_}, we obtain $s^{l \ast}(a)=x_{s+1}$.
	
	   All the other possible cases are the following:
	
	   \begin{enumerate}
		      \item $x_k < s^l(a)$: since $a \succ B_{s^l(a)}$ and after adding $b_{kp}$, Points $a.2.p$ and $a.2.d$ of
                    Condition \ref{cond:Separability} are still satisfied, then the highest $z$ such that $a \succ B_z$ remains $s^l(a)$, which, from Proposition \ref{prop:SoftDominancePreferenceAction_}, leads to $s^{l \ast}(a)=s^{l}(a)$;
		      \item  $x_k = s^l(a)$, but non$(b_{kp} \succ a)$: in this case, after adding the limiting profile $b_{kp}$,
                    we continue to have $a \succ B_{x_k}$, in such a way that the highest $z$ such that $a \succ B_z$ remains $s^l(a)$, which, from Proposition \ref{prop:SoftDominancePreferenceAction_}, leads to $s^{l \ast}(a)=s^{l}(a)$;
		      \item $x_k=x_{s+1}$, but not($a \succ b_{kp}$) or there is at least one $b_{kq} \in B_{x_{s+1}}$, such that,
                    $b_{kq} \succ a$: in this case, after adding the reference action $b_{kp}$, we continue to have not($a \succ B_{x_{s+1}}$), in such a way that the highest $z$, such that, $a \succ B_z$ remains $s^l(a)$, which, from Proposition \ref{prop:SoftDominancePreferenceAction_}, leads to $s^{l \ast}(a)=s^{l}(a)$;
		      \item $x_k>x_{s+1}$:  since not($a \succ B_{x_{s+1}}$) and since after adding $b_{kp}$, Points $a.2.p$ and
                    $a.2.d$ of Condition \ref{cond:Separability} are still satisfied, then, from Proposition \ref{prop:PropertiesActionsSets}, we obtain  not($a \succ B_{x_k}$), in such a way that the highets $z$ such that $a \succ B_z$ remains $s^l(a)$ which, by Proposition \ref{prop:SoftDominancePreferenceAction_}, leads to $s^{l \ast}(a)=s^{l}(a)$.
		\end{enumerate}
	\end{proof}

      \begin{lemma}[Inserting a reference action: consequences w.r.t. upper bound]\label{theo:InsertingSet}
         If $B$ fulfills both the primal and the dual soft dominance separability conditions (see Points $a.2.p$ and $a.2.d$ of Condition \ref{cond:Separability}), and if a limiting profile $b_{kp}$ is added to $B_{x_k}$ in such a way that these conditions are still fulfilled, and $s^u(a)=x_t$, then 				
				\begin{itemize}[label={--}]
					\item $s^{u \ast}(a)=x_{t+1}$, if $a \succ b_{kp}$ and $s^u(a)=x_t=x_k$;
					\item $s^{u \ast}(a)=x_{t-1}$, if $b_{kp} \succ a$, $x_k=x_{t-1}$ and there is
                    no $b_{kq} \in B_{x_{t-1}}$ such that $a \succ b_{kq} $;
					\item $s^{u \ast}(a)=s^u(a)$, otherwise.
				\end{itemize}
    \end{lemma}

    \begin{proof}
        If $a \succ b_{kp}$ and $s^u(a)=x_t=x_k$, then, after adding $b_{kp}$, we have not($B_{x_k} \succ a$).  Consequently, taking into consideration Proposition \ref{prop:SoftDominancePreferenceSet}, the lowest value $z$ such that $B_z \succ a$ becomes $x_{t+1}$, in such a way that $s^{u \ast}(a)=x_{t+1}$.
	
	    If $b_{kp} \succ a$, $x_k=x_{t-1}$ and there is no other $b_{kq} \in B_{x_{t-1}}$ such that $a \succ b_{kq}$, after
        adding $b_{kp}$, we  have $B_{x_{t-1}} \succ a$. In this case, the lowest value $z$ such that $B_z \succ a$ is no more $x_t$ and it becomes $x_{t-1}$, in such a way that, according to Proposition \ref{prop:SoftDominancePreferenceAction_}, we obtain $s^{u \ast}(a)=x_{t-1}$.
	
	    All the other possible cases are the following:
	
	  \begin{enumerate}
		  \item $x_k > s^u(a)$: since $B_{s^u(a)} \succ a$ and after adding $b_{kp}$,  Points $a.2.p$ and $a.2.d$ of
               Condition \ref{cond:Separability} are still satisfied, the lowest $z$ such that $B_z \succ a$ remains $s^u(a)$, which, by Proposition \ref{prop:SoftDominancePreferenceAction_}, leads to
               $s^{u \ast}(a)=s^{u}(a)$;
		\item  $x_k = s^u(a)$, but not$(a \succ b_{kp})$: in this case, after adding the reference action $b_{kp}$, we
            continue to have $B_{x_k} \succ a$, in such a way that the lowest $z$ such that $B_z \succ a$ remains $s^u(a)$, which, by Proposition \ref{prop:SoftDominancePreferenceAction_}, leads to $s^{u \ast}(a)=s^{u}(a)$;
		\item $x_k=x_{t-1}$, but not($b_{kp} \succ a$) or there is at least one $b_{kq} \in B_{x_{t-1}}$ such that $a
            \succ b_{kq}$: in this case, after adding the reference action $b_{kp}$, we continue to have not($B_{x_{t-1}} \succ a$), in such a way that the lowest $z$ such that $B_z \succ a$ remains $s^u(a)$, which, by Proposition \ref{prop:SoftDominancePreferenceAction_}, leads to $s^{u \ast}(a)=s^{u}(a)$;
		\item $x_k<x_{t-1}$:  since not($B_{x_{t-1}} \succ a$) and since after adding $b_{kp}$ Points $a.2.p$ and $a.2.d$
            of Condition \ref{cond:Separability} are still satisfied, then, by Proposition \ref{prop:PropertiesActionsSets}, we obtain  not($B_{x_k} \succ a$), in such a way that the lowest $z$ such that $B_z \succ a$ remains $s^u(a)$, which, by Proposition \ref{prop:SoftDominancePreferenceAction_}, leads to $s^{u \ast}(a)=s^{u}(a)$.
	  \end{enumerate}
	\end{proof}

    \begin{lemma}[Deleting a reference action: consequences w.r.t. lower bound]\label{theo:DeletingSet}
        If $B$ fulfills both the primal and the dual soft dominance separability conditions (see Points $a.2.p$ and $a.2.d$ of Condition \ref{cond:Separability}), and if a limiting profile $b_{kp}$ is removed from $B_{x_k}$, in such a way that, these conditions are still fulfilled, and $s^l(a)=x_s$, then				
				\begin{itemize}[label={--}]
					\item $s^{l \ast}(a)=x_{s-1}$, if $s^l(a)=x_k$, $a \succ b_{kp}$ and there is no other $b_{kq} \in
                        B_{x_k}$ such that $a \succ b_{kq}$,
					\item $s^{l \ast}(a)=x_{s+1}$, if $b_{kp} \succ a$, $x_k=x_{s+1}$, there is no other $b_{kq} \in
                        B_{x_{s+1}}$ such that $b_{kq} \succ a$ and there exists at least one reference action $b_{kr} \in B_{x_{s+1}}$ such that $a \succ b_{kr}$;
					\item $s^{l \ast}(a)=s^{l}(a)$, otherwise.
				\end{itemize}
    \end{lemma}

    \begin{proof}
        If $s^l(a)=x_k$, $a \succ b_{kp}$ and there is no other $b_{kq} \in B_{x_k}$ such that $a \succ b_{kq}$, then after $b_{kp}$ is removed we have not$(a \succ B_{x_k})$. Consequently, taking into consideration Proposition \ref{prop:SoftDominancePreferenceSet}, the highest value $z$ such that $a \succ B_z$ becomes $x_{s-1}$, in such a way that $s^{l \ast}(a)=x_{s-1}$.
	
        If $b_{kp} \succ a$, $x_k=x_{s+1}$, there is no other $b_{kq} \in B_{x_{s+1}}$ such that $b_{kq} \succ a$ and there exists at least one limiting profile $b_{kr} \in B_{x_{s+1}}$ such that $a \succ b_{kr}$, after removing  $b_{kp}$ we have $a \succ B_{x_{s+1}}$, in such a way that the highest $z$ such that $a \succ B_z$ becomes $x_{s+1}$ and, consequently, from Proposition \ref{prop:SoftDominancePreferenceAction_}, $s^{l \ast}(a)=x_{s+1}$.
	
        All the other possible cases are the following:
	
	\begin{enumerate}
		\item $x_k < s^l(a)$: since after removing $b_{kp}$ the highest $z$ such that $a \succ B_z$ remains $s^l(a)$ and
            Points $a.2.p$ and $a.2.d$ of Condition \ref{cond:Separability} are still satisfied, by Proposition \ref{prop:SoftDominancePreferenceAction_}, we obtain $s^{l \ast}(a)=s^{l}(a)$;
		\item  $x_k = s^l(a)$, and not($a \succ b_{kp}$) or there is at least another $b_{kq} \in B_{x_k}$ such that $a \succ b_{kq}$: in this case, after removing the limiting profile $b_{kp}$, we continue to have $a \succ B_{s^l(a)}$, and, therefore, $s^{l \ast}(a)=s^{l}(a)$;\
		\item $x_k=x_{s+1}$, but not($b_{kp} \succ a$) or there is at least another $b_{kq} \in B_{x_s+1}$ such that $b_{kq} \succ a$ or there does not exist any reference action $b_{kr} \in B_{x_{s+1}}$ such that $a \succ b_{kr}$: in this case, after removing the limiting profile $b_{kp}$, we continue to have not($a \succ B_{x_{s+1}}$), in such a way that the highest $z$ such that $B_z \succ a$ remains $s^l(a)$, which gives $s^{l \ast}(a)=s^{l}(a)$;
		\item $x_k>x_{s+1}$:  since not($a \succ B_{x_{s+1}}$) and since after removing $b_{kp}$ Points $a.2.p$ and $a.2.d$
            of Condition \ref{cond:Separability} are still satisfied, then, by Proposition \ref{prop:PropertiesActionsSets} we obtain  not($a \succ B_{x_k}$), in such a way that the highest $z$ such that $a \succ B_z$ remains $s^l(a)$, which leads to $s^{l \ast}(a)=s^{l}(a)$.
	\end{enumerate}
	
    \end{proof}

    \begin{lemma}[Deleting a reference action: consequences w.r.t. upper bound]\label{theo:DeletingSet}
        If $B$ fulfills both the primal and the dual soft dominance separability conditions (see Points $a.2.p$ and $a.2.d$ of Condition \ref{cond:Separability}) and if a limiting profile $b_{kp}$ is removed from $B_{x_k}$ in such a way that these conditions are still fulfilled, and $s^u(a)=x_t$, then				
				\begin{itemize}[label={--}]
					\item $s^{u \ast}(a)=x_{t+1}$, if $s^u(a)=x_k$, $b_{kp} \succ a$ and there is no other $b_{kq} \in
                        B_{x_k}$ such that $b_{kq} \succ a$,
					\item $s^{u \ast}(a)=x_{t-1}$, if $a \succ b_{kp}$, $x_k=x_{t-1}$, there is no $b_{kq} \in B_{x_{t-1}}$
                        such that $a \succ b_{kq}$ and there is at least one reference action $b_{kr} \in B_{x_{t-1}}$ such that $b_{kr} \succ a$;
					\item $s^{l \ast}(a)=s^{u}(a)$, otherwise.
				\end{itemize}
    \end{lemma}

    \begin{proof}
        If $s^u(a)=x_k$, $b_{kp} \succ a$ and there is no other $b_{kq} \in B_{x_k}$ such that $b_{kq} \succ a$, then after $b_{kp}$ is removed one has non$(B_{x_k} \succ a)$. Consequently, taking into consideration Proposition \ref{prop:SoftDominancePreferenceSet}, the lowest value $z$ such that $B_z \succ a$ becomes $x_{t+1}$, in such a way that $s^{u \ast}(a)=x_{t+1}$.

        If $a \succ b_{kp}$, $x_k=x_{t-1}$, there is no other $b_{kq} \in B_{x_{t-1}}$ such that $a \succ b_{kq}$ and there is at least one limiting profile $b_{kr} \in B_{x_{t-1}}$ such that $b_{kr} \succ a$, after removing  $b_{kp}$ one has $B_{x_{t-1}} \succ a$, in such a way that the lowest $z$ such that $a \succ B_z$ becomes $x_{t-1}$ and, consequently, from Proposition \ref{prop:SoftDominancePreferenceAction_},  $s^{u \ast}(a)=x_{t-1}$.
	
        All the other possible cases are the following:
		
        \begin{enumerate}
		      \item $x_k > s^u(a)$: since after removing $b_{kp}$ the lowest $z$ such that $B_z \succ a$ remains $s^u(a)$
                    and  Points $a.2.p$ and $a.2.d$ of Condition \ref{cond:Separability} are still satisfied, by Proposition \ref{prop:SoftDominancePreferenceAction_},  we obtain $s^{u \ast}(a)=s^{u}(a)$;
			 \item  $x_k = s^u(a)$, and not($b_{kp} \succ a$) or there is at least another $b_{kq} \in B_{x_k}$ such that
                    $b_{kq} \succ a$: in this case, after removing the reference action $b_{kp}$, we continue to have $B_{s^u(a)} \succ a$, and, therefore, $s^{u \ast}(a)=s^{u}(a)$;
			\item $x_k=x_{t-1}$, but not($a \succ b_{kp}$) or there is at least another $b_{kq} \in B_{x_{t-1}}$ such that
                $a \succ b_{kq}$, or there does not exist any reference action $b_{kr} \in B_{x_{t-1}}$ such that $b_{kr} \succ a$: in this case, after removing the reference action $b_{kp}$, we continue to have not($B_{x_{t-1}} \succ a$), in such a way that the lowest $z$ such that $B_z \succ a$ remains $s^u(a)$, which leads to $s^{u \ast}(a)=s^{u}(a)$;
		\item $x_k<x_{t-1}$:  since not($B_{x_{t-1}} \succ a$) and since after removing $b_{kp}$ Points $a.2.p$ and $a.2.d$
             of Condition \ref{cond:Separability} are still satisfied, then, by Proposition \ref{prop:PropertiesActionsSets}, we obtain  not($B_{x_k} \succ a$), in such a way that the lowest $z$, such, that $B_z \succ a$ remains $s^u(a)$, which leads to $s^{u \ast}(a)=s^{u}(a)$.
			\end{enumerate}
	\end{proof}

    Putting together the results of Lemmas \ref{theo:InsertingSet_}-\ref{theo:DeletingSet} we get the following general result.

    \begin{theorem}[Stability]\label{theo:MainTheorem}
         If $B$ fulfills both the primal and the dual soft dominance separability conditions (see Points $a.2.p$ and $a.2.d$ of Condition \ref{cond:Separability}) then the stability condition holds.
    \end{theorem}

\section{Practical issues and an illustrative example}\label{sec:PracticalIssuesExample}
\noindent This section presents some practical aspects of the method and an illustrative example

\subsection{Building the set of reference actions}\label{sec:SetRefActions}
\noindent Let us start by observing that in defining the collection of sets of references actions
$$B = \{B_{x_1},\ldots,B_{x_k},\ldots,B_{x_\ell}\}$$
 it is necessary to take care that all actions $a$ from $A$ the following condition is satisfied
\begin{equation}\label{comparability}
B_{x_\ell} \succ a \succ B_{x_1}.
\end{equation}
Condition (\ref{comparability}) permits to compare all actions $a$ from $A$ with sets $B_x$ of reference actions, so that a range $]s^l(a), s^u(a)[$ for the score $s(a)$ is assigned to each $a \in A$.

\noindent There are different ways of building the set of reference actions. Next we present two different ones.

\paragraph{Direct method}
We ask the decision-maker to propose a set of reference actions that are used to characterize some scoring levels, for example, $20$, $40$, $60$, and $80$. Several reference actions can be proposed \emph{a priori} to characterize the same scoring level. We assume that the scales are bounded from below and from above in such a way that we can consider reference the actions $b_0$ and $b_{100}$, characterizing the scores of $0$ and $100$ with actions having, on each considered criterion, the lower bound performances and the upper bound performances, respectively.

\paragraph{A deck of cards based technique} We can also ask the decision-maker to propose a set of actions to be taken as reference actions and use the deck of cads method in a similar way as in \cite{FigueiraRoy02}. These actions must be ordered from the worst to the best, with the possibility of some ties. The decision-maker is required to put some blank cards between these equivalent sets of reference actions in order to assign a score to each equivalence class, on the basis of the extreme scores $0$ and $100$. The scores thus obtained are not necessarily equally spaced. It should be noticed that other values different from $0$ and $100$ can be used (this is also true when using the previous method).

\subsection{The method}\label{sec:Method}
\noindent The method is now presented in a very simple way. Remember that in order to assign a score $s(a)$ to the action $a \in A$, we should firstly choose a $\lambda$ cutting level to transform the fuzzy outranking relation into a crispy one.

\begin{enumerate}
    \item Find the set of reference actions, $B_x$, with the highest score $x$, such that $a \succ B_x$ and, for all $x^\prime < x$, we have either $a \succ B_{x^\prime}$ or $a \parallel B_{x^\prime}$. It is then natural to consider a score $s(a)$ such that $s(a) > x= s^l(a)$, and indeed, according Definition \ref{def:LowerBound}, we have $s^l(a)=x$. Let us remember also that by Proposition \ref{prop:SoftDominancePreferenceAction_}, if both the primal and the dual soft dominance separability conditions (see Condition \ref{cond:Separability}, points   $a.2.p$ and $a.2.d$) hold, then $s^l(a)$ is simply the highest value, such that $a\succ B_x$.
		
		\item Find the set of reference actions, $B_x$, with the lowest score $x$, such that $B_x \succ a$ and, for all $x^\prime > x$, we have either $B_{x^\prime} \succ a$ or $a \parallel B_{x^\prime}$. It is then natural to consider a score $s(a)$ such that $s(a) < x= s^u(a)$, and indeed, according Definition \ref{def:UpperBound}, we have $s^u(a)=x$. Let us remember also that by Proposition \ref{prop:SoftDominancePreferenceAction_}, if both the primal and the dual soft dominance separability conditions (see Condition \ref{cond:Separability}, points   $a.2.p$ and $a.2.d$) hold, then $s^u(a)$ is simply the lowest value, such that $B_x \succ a$.
\end{enumerate}

\subsection{Illustrative example}\label{sec:Example}
\noindent  We consider the example in \cite{FigueiraEtAl09} regarding the evaluation of some sites for the location of a new hotel. The set of criteria is as follows:

\begin{enumerate}
  \item \textit{Investment costs} (Scale unit: $K$\euro; Code: {\tt ICOST}; notation: $g_1$; preference direction: minimization). This criterion comprises the land purchasing costs as well as the costs for building the new hotel.
  \item \textit{Annual costs} (Scale unit: $K$\euro; Code: {\tt ACOST}; notation: $g_2$; preference direction: minimization). This criterion comprises the hotel operating costs.
  \item \textit{Recruitment} (Scale unit: verbal levels (seven); Code: {\tt RECRU}; notation: $g_3$; preference direction: maximization). This criterion models the possibility of recruiting workers.
  \item \textit{Image} (Scale unit: verbal levels (seven); Code: {\tt IMAGE}; notation: $g_4$; preference direction: maximization). This criterion models the perceptions of the clients about the district where the new hotel will be located.
  \item \textit{Access} (Scale unit: verbal levels (seven); Code: {\tt ACCES}; notation: $g_5$; preference direction: maximization). This criterion models the possibility of recruiting workers.
\end{enumerate}

The verbal scale used for the last three criteria comprises the following levels  (in between parenthesis we used a numerical code for each level): {\tt very bad}[1]; {\tt bad}[2]; {\tt rather bad}[3]; {\tt average}[4]; {\tt rather good}[5]; {\tt good}[6]; {\tt very good}[7].

There are five potential sites for the location of the new hotel. The performance table can be presented as follows (Table 1).

\begin{table}[htb!]
  \centering
    \begin{tabular}{cccccc}\hline
$a$ & {\tt ICOST} ($g_1$) & {\tt ACOST} ($g_2$) & {\tt RECRU} ($g_3$) & {\tt IMAGE} ($g_4$) & {\tt ACCES} ($g_5$)   \\ \hline
  $a_1$ & $13\;000$ & $3\;000$ & 4 & 4 & 4 \\
  $a_2$ & $15\;000$ & $2\;500$ & 6 & 2 & 7 \\
  $a_3$ & $10\;900$ & $3\;400$ & 6 & 6 & 1 \\
  $a_4$ & $15\;500$ & $3\;500$ & 6 & 6 & 6 \\
  $a_5$ & $15\;000$ & $2\;600$ & 6 & 1 & 2 \\ \hline
    \end{tabular}
  \caption{Performance table}\label{tab:Perf_Table}
\end{table}

The weights, discriminating (indifference and preference) thresholds used in the method are the following (see Table 2).  Let us consider an ordered pair of actions $(a,b) \in A \times A$. The performance of $b$ is assumed to be worse than the performance of $a$. This means that the variable thresholds presented in the next table with respect to criteria $g_1$ and $g_2$ are direct variable thresholds \citep[see][]{RoyEtAl14}. For the sake of simplicity, no veto thresholds are considered in this example. For the remaining criteria the thresholds are constant (the numbers represent the differences of levels, not the scale levels. Thus, the indifference threshold is a difference of one performance level, while the preference threshold corresponds to a difference of two performance levels).

\begin{table}[htb!]
  \centering
    \begin{tabular}{cccccc}\hline
  Parameters & {\tt ICOST} ($g_1$) & {\tt ACOST} ($g_2$) & {\tt RECRU} ($g_3$) & {\tt IMAGE} ($g_4$) & {\tt ACCES} ($g_5$)   \\ \hline
  $k_j$         & 5 & 4 & 3 & 3 & 3 \\ \hline
  $q_j(g_j(b))$ & $250+0.03g_1(b)$ &  $\;\,50+0.05g_2(b)$ & 1 & 1 & 1 \\
  $p_j(g_j(b))$ & $500+0.05g_1(b)$ & $100+0.07g_2(b)$ & 2 & 2 & 2 \\ \hline
    \end{tabular}
  \caption{Parameters table}\label{tab:Parameters_Table}
\end{table}

Our set of limiting profiles is composed of seven subsets, i.e., $B = \{B_{x_1}, B_{x_2}, B_{x_3},B_{x_4},B_{x_5},B_{x_6},$ $B_{x_7}\}$. This means that will be defined seven reference values $X=\{x_1,x_2,x_3,x_4,x_5,x_6, x_7 \}$. These values were obtained with the deck of cards method proposed in \cite{BotteroEtAl18} to build interval scales, by fixing two reference levels. In this work, and without loss of generality, we will put $x_1=0$ and $x_7=100$.

The $B$ sets are characterized by at least a single limiting profile as follows.

\begin{itemize}
  \item[~] $B_{x_1=0} = \{b_{11} = (18\,000, 4\,000, 1, 1, 1)\}$;
  \item[~] $B_{x_2} = \{b_{21} = (17\,000, 3\,500, 2, 2, 1), \; b_{22} = (16\,500, 3\,700, 1, 2, 1)\}$;
  \item[~] $B_{x_3} = \{b_{31} = (15\,350, 3\,200, 3, 1, 2)\}$;
  \item[~] $B_{x_4} = \{b_{41} = (14\,250, 2\,850, 3, 4, 3), \; b_{42} = (13\,750, 3\,150, 4, 3, 3)\}$;
  \item[~] $B_{x_5} = \{b_{51} = (12\,650, 2\,650, 4, 4, 5)\}$;
  \item[~] $B_{x_6} = \{b_{61} = (11\,500, 2\,100, 5, 6, 5), \; b_{62} = (11\,000, 2\,500, 6, 5, 7)\}$;
  \item[~] $B_{x_7=100} = \{b_{71} = (10\,000, 2\,000, 7, 7, 7)\}$;
\end{itemize}

The deck of cards method applied to this problem works as follows.

\begin{enumerate}
  \item The subsets of reference profiles are totally ordered as follows ($\prec$ means ``strictly less preferred than''):
    \[
        B_{x_1} \prec B_{x_2} \prec B_{x_3} \prec B_{x_4} \prec B_{x_5} \prec B_{x_6} \prec B_{x_7}.
    \]
  \item We then call the attention of the decision-maker to the fact that if the difference between two consecutive sets is bigger than the difference of other pair of consecutive subsets, she/he should add more blank cards in the first two consecutive ones than in the second ones \citep[for more details see][]{BotteroEtAl18}. We may obtain the following ranking of the subsets with the blank cards in between consecutive ones (in brackets):
       \[
        B_{x_1} \; [1] \; B_{x_2} \; [2] \; B_{x_3} \; [0] \; B_{x_4} \; [1] \; B_{x_5} \; [0] \; B_{x_6} \; [2] \; B_{x_7}.
        \]
    Zero blank cards between two consecutive sets of limiting profiles does not mean that the limiting profiles of the two consecutive sets have the same value, but only that the difference is minimal. We know that the number of units between $B_{x_1}$ and $B_{x_7}$ is $\alpha = (1+1) + (2+1) + (0+1) + (1+1) + (0+1) + (2+1) = 12$.
  \item In this step we should define the two reference levels. As stated before, we considered that the value of all the limiting profiles in $B_{x_1}$ is $0$, i.e., $x_1=0$ and that the value of all the limiting profiles in $B_{x_7}$ is $100$, i.e., $x_7=100$.
  \item Consequently, we can compute the value of the unit as follows.
    \[
        u = \frac{x_7 - x_1}{h} = \frac{100-0}{12}= 8.333333.
    \]
  \item The computations of the remaining values in $X$ is easy since we know the number of units separating two consecutive sets of limiting profiles. We have thus,
      \[
            x_1=0, \; x_2= 25, \; x_3= 33.33333333, \; x_4 = 50, \; x_5 = 58.33333333, \; x_6 = 83.33333333,\; x_7=  100.
      \]
\end{enumerate}

Let us consider now the comparison table of all of our five sites against the limiting profiles (Table 3).

\begin{table}[htb!]
  \centering
    \begin{tabular}{|c|ccccc|ccccc|}\hline
    $b$    & $a_1\succ b$ &  $a_2\succ b$ & $a_3\succ b$ & $a_4\succ b$ & $a_5\succ b$ & $b \succ a_1 $ &  $b \succ a_2$ & $b \succ a_4$ & $b \succ a_4$ & $b \succ a_5$ \\ \hline
    $b_{11}$ & $\succ$ & $\succ$ & $\succ$ & $\succ$ & $\succ$ &         &         &         &         & \\ \hdashline
    $b_{21}$ & $\succ$ & $\succ$ & $\succ$ & $\succ$ & $\succ$ &         &         &         &         & \\
    $b_{22}$ & $\succ$ & $\succ$ & $\succ$ & $\succ$ & $\succ$ &         &         &         &         & \\ \hdashline
    $b_{31}$ & $\succ$ & $\succ$ & $\succ$ & $\succ$ & $\succ$ &         &         &         &         & \\ \hdashline
    $b_{41}$ &         & $\succ$ &         & $\succ$ & $\succ$ &         &         &         &         & \\
    $b_{42}$ &         & $\succ$ & $\succ$ &         &         &         &         &         &         & \\ \hdashline
    $b_{51}$ &         &         &         &         &         &         &         &         & $\succ$ & $\succ$ \\ \hdashline
    $b_{61}$ &         &         &         &         &         & $\succ$ & $\succ$ & $\succ$ & $\succ$ & $\succ$ \\
    $b_{62}$ &         &         &         &         &         & $\succ$ & $\succ$ & $\succ$ & $\succ$ & $\succ$ \\ \hdashline
    $b_{71}$ &         &         &         &         &         & $\succ$ & $\succ$ & $\succ$ & $\succ$ & $\succ$ \\  \hline
    \end{tabular}
  \caption{Comparison table}\label{tab:Perf_Table}
\end{table}

Now, for the definition of the score range of each alternative we can take advantage of this table. Let us consider action $a_1$ and try to identify $s^l(a_1) < s(a) < s^u(a_1)$. As for the lower bound it is provided by $x_3$ and the upper bound by $x_6$. Thus, the range for score of $a_1$ becomes.

\[
    33.33333 < s(a_1) < 83.33333.
\]

With the same procedure we can derive the range for all the five actions:

\begin{itemize}
  \item[~] $33.33333 < s(a_1) < 83.33333$;
  \item[~] $50.00000 < s(a_2) < 83.33333$;
  \item[~] $50.00000 < s(a_3) < 83.33333$;
  \item[~] $33.33333 < s(a_4) < 58.33333$;
  \item[~] $33.33333 < s(a_5) < 58.33333$;
\end{itemize}

%

\section{Conclusions}\label{sec:Conclusion}
\noindent In this paper we presented a new method for assigning a score range to each action. In this sense it can provide a more robust conclusion about the possible scores such an action can have.  The theoretical soundness of the method has been proved, i.e., the fundamental requirements of uniqueness, independence, monotonicity, conformity, homogeneity, and stability with respect to insertion and deletion operations, became properties of the method, which provide some consistency to  the method.

All the main strengths of {\sc{Electre}} methods are present in this new method: it deals with different types of scales without the need of converting them into a single unit; it is able to cope with the imperfect knowledge of data and arbitrariness when building the criteria; it takes into account the reasons for and against an outranking; and, it avoids the compensatory phenomenon in a systematic way. In addition, the method is able to provide a score range for each action, which, before was considered a weak point of {\sc{Electre}} methods.

Future research lines, start from software development and real-work applications. They include also several extensions, as for example,
\begin{itemize}
	\item the design of a hierarchical {\sc{Electre-Score}} method, as in \cite{CorrenteEtAl13, CorrenteEtAl16} and \cite{DelVastoTeriientesEtAl15}, allowing in our case to give ranges of scores for different macrocriteria in the hierarchical trees of criteria;
	\item the application of Monte-Carlo based pseudo-robustness analysis in the same line as in \cite{CorrenteEtAl17} to provide more accurate score ranges;
	\item the use of robust ordinal techniques and other robustness techniques in the same line of \cite{GrecoEtAl10} and  \cite{KadzinskiCi16};
	\item the development evolutionary approaches for making the computations when in presence of large sets of data as in \cite{DoumposEtAl09}.
\end{itemize}

 \vspace{0.5cm}

\section*{Acknowledgements}
\addcontentsline{toc}{section}{\numberline{}Acknowledgements}
\noindent Jos\'e Rui Figueira acknowledges the support  the Funda\c c\~ao para a Ci\^encia e Tecnologia (FCT) grant SFRH/BSAB/139892/2018 under POCH program and from the European Regional Development Fund (ERDF) as well as FCT for supporting the DECIdE project (LISBOA-01-0145-FEDER-024135), and for partially funding the current research. Salvatore Greco wishes to acknowledge funding by the FIR of the University of Catania BCAEA3 ``New developments in Multiple Criteria Decision Aiding (MCDA) and their application to territorial competitiveness''.


\vfill\newpage

~~

\vspace{-3.00cm}

\renewcommand{\thesection}{\Alph{section}}
\setcounter{section}{0}
\section{Appendix}\label{sec:Appendix_Parameters}
\noindent In this Appendix we will present the main concepts and steps that lead to the construction of a credibility degree, $\sigma(a,b)$ for the pair of actions $(a,b)$.

State-of-the art versions of {\sc{Electre}} methods make use of the so-called \emph{pseudo-criterion} model \citep{FigueiraEtAl16,Roy96,RoyBo93,RoyVi84} when comparing two actions, $a$ and $b$, on criterion $g_j$, from their performances, $g_j(a)$ and $g_j(b)$, respectively. This model associates with each criterion function, $g_j(\cdot)$, two threshold functions: an \emph{indifference threshold} function, denoted by $q_j(\cdot)$, and a \emph{preference threshold} function, denoted by $p_j(\cdot)$. Assume that $g_j$ is a criterion to be maximized and that the performance $g_j(a)$ is better than the performance $g_j(b)$. The threshold functions or simply thresholds can be constant or may vary in a direct way, i.e., with respect to the worst performance, $q_j(g_j(b))$ and $p_j(g_j(b))$, or in an inverse way, i.e., with respect to the best performance, $q_j(g_j(a))$ and $p_j(g_j(a))$. For the sake of simplicity and without loss of generality, we consider in what follows that the thresholds are constant and use the simple notation $q_j$ and $p_j$, for the indifference and preference thresholds, respectively.

It is very important to note that the main purpose of these thresholds is not to model the preferences, but rather the imperfect knowledge of data as it can be seen in \cite{RoyEtAl14}.

The choice of a pseudo-criterion model for the comparison of two actions, $a$ and $b$, from their performances on criterion $g_j$, leads to the definition of three \emph{per}-criterion binary relations, as follows.

\begin{itemize}[label={--}]
    \item A per-criterion \emph{indifference binary relation}, which is used to model a situation in which $a$ is indifferent to $b$ on criterion $g_j$, denoted by $a\sim_j b$; this occurs whenever $\vert (g_j(a) - g_j(b)) \vert \leqslant q_j$. In other words, a situation where no one of the two actions, $a$ and $b$, has a significant advantage over the other on the considered criterion. Let $C(a\sim b)$ denote the set or coalition of criteria for which $a$ is indifferent to $b$.
    \item  A per-criterion \emph{strict preference binary relation}, which is used to model a situation in which $a$ is strictly preferred to $b$ on criterion $g_j$, denoted by $a\succ_j b$; this occurs whenever $(g_j(a) - g_j(b)) > p_j$. In other words, a situation where action $a$ has a significant advantage over $b$ on the considered criterion. Let $C(a\succ b)$ denote the set or coalition of criteria for which $a$ is strictly preferred to $b$.
    \item A per-criterion \emph{weak preference binary relation}, which is used to model hesitation situations of $a$ with respect to $b$ on criterion $g_j$, denoted by $a\succsim^{?}_{j} b$; this occurs whenever $q_j < (g_j(a) - g_j(b)) \leqslant p_j$. In other words, a situation where there is an ambiguity zone between indifference and strict preference of $a$ over $b$ on the considered criterion. Let $C(a\succsim^{?} b)$ denote the set or coalition of criteria for which $a$ is weakly preferred to $b$. Note that the word weak has nothing to do with intensities of preferences, it models hesitation or ambiguity (due to the imperfect knowledge of data), not preferences.
\end{itemize}

Whenever $a$ is indifferent, weak, or strict preferred to $b$, on criterion $g_j$, we say that ``$a$ outranks $b$'' since $a$ is at least as good as $b$ in a \emph{stricto sensu} on this criterion. This situation thus occurs, when $a\sim_j b$, $a\succ_j b$, or $a\succsim^{?}_j b$, and can be denoted by $a\succsim_j b$. In a more \emph{lato sensu} we can also say that ``$a$ outranks $b$'' when $b\succsim^{?}_j a$ since there is hesitation between $b\sim_j a$ and $b\succ_j a$ on criterion $g_j$.

As all outranking based methods, {\sc{Electre}} methods also make use of the \emph{per}-criterion outranking relations to build one or several comprehensive outranking relations. In the method proposed in this paper only one comprehensive outranking relation is considered, which allows to conclude whether or not ``$a$ comprehensively outranks $b$'', denoted by $a\succsim b$.  More precisely, to conclude about the assertion ``$a$ outranks $b$'', the strength of the coalition in its favor of should be powerful enough to overcome the opposition effect of the coalition against this assertion. How should the power of the coalition in favor (or concordant with the assertion) and the effects of the coalition against (or discordant with the assertion) be measured?  To model the power of the concordant coalition is modeled and measured in {\sc{Electre}} methods through what is called in these methods a \emph{comprehensive concordance index}, while the opposition effect of each criterion is modeled and measured to what is called a \emph{per-criterion discordance index}. Both will be combined to devise a \emph{credibility} (outranking) \emph{index} for each ordered pair of alternatives, $(a,b) \in A\times A$. Next, we will present the three main steps to obtain the credibility index.

\begin{enumerate}
    \item Computing the \emph{comprehensive concordance index} $c(a,b)$. Again, for the sake of simplicity, the formula we present next, for this index, is the classical one as in \cite{RoyBo93}. A more sophisticated and recent version of the concordance index, which takes into account the interaction between criteria can also be used in this context \citep[see][]{FigueiraEtAl09} with no additional changes in the method proposed in this paper. As said before the concordance index is used for measuring the power of the concordant coalition, where each criterion $g_j$ contributes with its relative importance coefficient or weigh, $w_j$, for $j=1,\ldots,n$ (we assume w.l.o.g. that $\sum_{j=1}^nw_j=1$). The formula for the index can thus be stated as follows.
        \[
            {\displaystyle c(a,b) = \sum_{C(a\{\sim,\succsim^?,\succ\}b)}w_j + \sum_{C(b\succsim^? a)}\varphi_jw_j},
        \]
        where
        \[
            {\displaystyle \varphi = \frac{(g_j(a) - g_j(b)) + p_j}{p_j - q_j} \in [0,1]}.
        \]
        This means that if a criterion $g_j$ belongs to the concordant coalition \emph{stricto sensu}, i.e., $g_j \in C(a\{\sim,\succsim^?,\succ\}b$ its contribution to the coalition power corresponds to its total weight, $w_j$, but if this criterion belongs to $C(b\succsim^? a)$ it only contributes with a fraction of its weight, $\varphi_jw_j$.
    \item Computing the \emph{per-criterion discordance indices} $d_j(a,b)$, $j=1,\ldots,n$. To model the opposition effect of each criterion against the concordant coalition \emph{lato sensu}, i.e., when $g_j \in C(b\succ a)$, it is necessary to introduce another concept and preference parameter, the \emph{veto threshold} $v_j(\cdot)$. This threshold can also be constant or vary in a direct or indirect way as in case of indifference and preference thresholds. To render things simple and without loss of generality we will keep its value constant, and denoted it simply by $v_j$. A criterion $g_j$ is discordant with the assertion ``$a$ outranks $b$'', when the difference of performances $(g_j(b) - g_j(a))$ is considered significantly large to validate such an assertion.   The more or less degree of discordance of each criterion can be measured through a \emph{per}-criterion discordance  index of the form,
         \[
            d_j(a,b) =
            \left\{
            \begin{array}{ccl}
            1                         & \mbox{if} & \;\;\,v_j > (g_j(a) - g_j(b))  \\
         {\displaystyle \frac{(g_j(a) - g_j(b)) + p_j}{p_j - v_j}} & \mbox{if} & -v_j \leqslant (g_j(a) - g_j(b)) < -p_j\\
            0                                               & \mbox{if} & \hspace{1.15cm}(g_j(a) - g_j(b)) \geqslant -p_j \\
            \end{array}
            \right.
        \]
    \item Computing the \emph{credibility index} $\sigma(a,b)$. This index measures the credibility degree of the outranking relation, i.e., the degree in which $a$ outranks $b$. It can be modeled though the following formula.
        \[
            {\displaystyle \sigma(a,b) = c(a,b)\prod_{j=1}^{n}T_j(a,b)},
        \]
        where
        \[
             T_j(a,b) =
             \left\{
             \begin{array}{cl}
                {\displaystyle  \frac{1 - d_j(a,b)}{1 - c(a,b)}} & \mbox{if} \;\, g_j(a,b) > c(a,b) \\
                 1 & \mbox{otherwise} \\
             \end{array}
             \right.
        \]
        It is thus a fuzzy measure. It can be converted into a crispy by making use of a cutting-off level, denote by $\lambda$ as in Section \ref{sec:Concepts}.
\end{enumerate}

For the main features, advantages, and drawbacks of {\sc{Electre}} the reader can refer to \cite{FigueiraEtAl13}.

\vspace{-0.25cm}


\noindent

\vfill\newpage

\bibliographystyle{model2-names}
\bibliography{Bib_Score}
\addcontentsline{toc}{section}{References}







\end{document}